\documentclass[fontsize=10pt]{article}
\usepackage[hyphens]{url}
\usepackage[a4paper, total={6in, 8in}]{geometry}
\usepackage{graphicx}
\usepackage{natbib}
\usepackage[font=small,labelfont=bf]{caption}

\usepackage{mathtools}
\usepackage{amsthm}
\usepackage{amsmath}
\usepackage{amsfonts}
\usepackage{amssymb}

\usepackage[hidelinks]{hyperref}
\usepackage{cleveref}

\usepackage{authblk}

\usepackage{graphicx}
\usepackage{amsmath}
\usepackage{amsthm}
\usepackage{booktabs}
\usepackage{algorithm}
\usepackage[noend]{algorithmic}

\usepackage[switch]{lineno}

\usepackage{amsfonts}
\usepackage{color-edits}
\usepackage{todonotes}
\usepackage{thmtools, thm-restate}

\DeclareMathOperator{\prob}{\textsc{MinHazing}}
\DeclareMathOperator{\welfprob}{\textsc{MaxWelfareMinHazing}}
\DeclareMathOperator{\poly}{poly}

\newcommand{\np}{\textsc{NP}}
\newcommand{\Z}{\mathbb{Z}}

\newcommand{\N}{\mathbb{N}}
\newcommand{\mem}{\texttt{memo}}
\newcommand{\none}{\texttt{none}}
\newcommand{\bd}{B}

\newtheorem{theorem}{Theorem}
\newtheorem{example}{Example}

\newtheorem{proposition}[theorem]{Proposition}
\newtheorem{corollary}[theorem]{Corollary}

\theoremstyle{definition}
\newtheorem{definition}[theorem]{Definition}
\theoremstyle{definition}

\newtheorem{remark}[theorem]{Remark}

\title{Beyond Symmetry in Repeated Games with Restarts}

\author[1]{Henry Fleischmann}
\author[1]{Kiriaki Fragkia}
\author[1,2]{Ratip Emin Berker}

\affil[1]{Carnegie Mellon University}
\affil[2]{Foundations of Cooperative AI Lab (FOCAL)}
\affil[ ]{\texttt{\{hfleisch, kiriakif, rberker\}@cs.cmu.edu}}

\begin{document}

\maketitle

\begin{abstract}
Infinitely repeated games support equilibrium concepts beyond those present in one-shot games (\emph{e.g.}, cooperation in the prisoner's dilemma). Nonetheless, repeated games fail to capture our real-world intuition for settings with many anonymous agents interacting in pairs. Repeated games with restarts, introduced by Berker and Conitzer [IJCAI '24], address this concern by giving players the option to restart the game with someone new whenever their partner deviates from an agreed-upon sequence of actions. In their work, they studied symmetric games with symmetric strategies.  We significantly extend these results, introducing and analyzing more general notions of equilibria in asymmetric games with restarts. We characterize which goal strategies players can be incentivized to play in equilibrium, and we consider the computational problem of finding such sequences of actions  with minimal cost for the agents. We show that this problem is $\np$-hard in general. However, when the goal sequence maximizes social welfare, we give a pseudo-polynomial time algorithm.
\end{abstract}

\section{Introduction}

\par Social dilemmas often arise when individuals aim to satisfy their own incentives, which often may prohibit cooperation. In fact, in many games, although cooperating could yield better payoffs for both players, it does not yield a Nash equilibrium and is hence unlikely to occur. Repeated games can circumvent this concern by capturing more complex and realistic notions of equilibria, where mutual cooperation can be incentivized. For example, consider the game in Table \ref{symmetricgametable}.

{\begin{table}[b]
\centering

\begin{tabular}{| c | c |c| c | }
\hline
 {} & $C_1$ & $C_2$ & $D$   \\
 \hline
 $C_1$ & $8, 8$ & $0, 8$ & $0, 17$\\  
 $C_2$ & $8, 0$ & $2, 2$ & $0, 11$\\
 $D$  & $17, 0$ & $11, 0$ & $1, 1$ \\
 \hline
\end{tabular}

\caption{Symmetric repeated game}
\label{symmetricgametable}
\end{table}}

Notice that in the single-shot version of the game, actions $C_1$ and $C_2$ (``cooperate'') are strictly dominated by action $D$ (``defect''). This leads to $(D, D)$ being a dominant strategy Nash equilibrium, even though each player could obtain more value by cooperating. However, when playing this game repeatedly, cooperation \emph{ad infinitum} can be an equilibrium given that players are sufficiently patient: both players can agree to cooperate by playing $C_1$ until their opponent defects, at which point they start playing action $D$ forever~\cite{grimtrig}. In this case, no player is incentivized to deviate from $C_1$, since any additional payoff they could receive from deviating would be offset by the subsequent punishment.

However, this type of collaboration fails in many real-world anonymous settings, in which players can choose to leave the game and restart with someone new. Consider an infinite collection of agents playing a repeated game in pairs either forever or until one of the players chooses to leave, in which case they are assigned a new partner. If there is no way for a player to check their new partner's history, a malicious agent could hop from partner to partner, defecting and then immediately leaving before suffering any punishment. In real-life relationships, such as ones between colleagues, freelancers with clients, or among romantic partners, agents tend to more gradually build up trust to avoid repeated exploitation. How can we formalize this game-theoretically?

One way is to consider repeated games \emph{with restarts} \cite{berker2024computing}, in which pairs of anonymous agents play an infinitely repeated game with the option to restart the game with a new player at any point. Consider then the strategy of everyone agreeing on a common sequence of actions to take, and if either player in a pair ever deviates from the sequence, the sequence is restarted. This simulates the agent punishing a defector by leaving the relationship and seeking a new partner. Ideally, such a sequence would incentivize agents to follow it at the risk of initiating a relationship with a new partner, which might come at a high cost.

Concretely, consider again the game in Table \ref{symmetricgametable} and let $(D, D), (C_2, C_2), (C_1, C_1),(C_1, C_1), \dots$ be a sequence of action pairs that both players commit to playing. Notice that the first action pair $(D, D)$ is a dominant strategy Nash equilibrium, so no player has an incentive to deviate. In the second round, an agent can guarantee an additional payoff of $+9$ by deviating to action $D$. However, this results in their partner ending the relationship, at which point the deviating player will have to restart the sequence. This results in a $(1+11)/2 = 6$ per-round average payoff, whereas that player could have eventually guaranteed an average payoff of 8 by choosing to follow the sequence as is. Deviating to $D$ on round three yields an additional payoff of $9$, but this only amounts to a $(1 + 2 + 17)/3 \approx 6.66$ average, compared to the $8$ they could have gained following the existing sequence. A similar reasoning applies for future rounds, ensuring stability.

\par In previous work, \citet{berker2024computing} formalize this subclass of Nash equilibria in repeated games with restarts and analyze its computational complexity, in the restricted setting of symmetric games and symmetric strategies. However, even in a simple example such as the one above, having players alternate between two actions (an asymmetric strategy) can yield a higher per-round average payoff. For example, say from the third round onward, players follow the strategy $(C_1, D), (D, C_1), \ldots$ (alternating between actions $C_1$ and $D$). Then, each will receive a per-round average payoff of $(0+17)/2 = 8.5$, compared to the $8$ that the best symmetric strategy $(C_1, C_1)$ could yield. This sequence is also stable: when each player plays action $D$, they receive a payoff of $17$ and have no incentive to deviate. When playing $C_1$, they could deviate to $D$ for a $+1$ additional payoff, but this is once again offset by the cost of restarting the sequence. This shows that, \emph{even in symmetric games}, equilibria with asymmetric strategies improve outcomes for both players. Therefore, in our work we aim to answer the following:

\begin{center}
\emph{How can we optimize payoff of a (possibly asymmetric) equilibrium sequence in (possibly asymmetric) repeated games with restarts? 
}
\end{center}

\subsection{Related Work}

\par Infinitely repeated games \emph{without} restarts are well studied in the literature. For a thorough treatment, see \citet{mailath2006repeated} and \citet{mertens2015repeated}. In particular, there are numerous characterizations of equilibria, referred to as Folk Theorems (see, for example, \citet{grimtrig} and \citet{fudenberg86}). One interpretation of a Folk theorem is that, for each action pair $(a^{(1)}, a^{(2)})$ where players receive strictly more utility than their minmax payoff, there is a strategy and a sufficiently large discount factor such that $(a^{(1)}, a^{(2)})$ is repeated forever in equilibrium. (Recall that in repeated games it is typical to introduce a discount factor $\beta \in (0,1)$ such that the round $i$ utility is scaled by $\beta^i$.) Here, the minmax payoff refers to the maximum payoff a player gets if their opponent plays the action minimizing the first player's maximum payoff. The key idea is that either player can punish their opponent for deviating by playing the action minimizing the opponent's (maximum) utility. 
 
The Folk theorem result most relevant to our work is that of \citet{fudenberg86}. In their setting, the mere threat of punishment motivates players to adhere to Nash equilibria, since leaving your partner is not allowed. In comparison, as we will see, agents in our setting must be \emph{hazed} upfront to prevent serial defectors. This distinction arises as a result of our model capturing anonymity among players, a feature common to many interactions in the real world. Another difference between our works is that our focus is not whether equilibria with a given stable sequence exist (the direct analogue of typical Folk theorem guarantees), but, given that they do, we aim to find the ``best'' such equilibrium among them. We view this as finding the ``least severe'' punishment for deviation that still ensures an equilibrium.

The negative impacts of anonymity on establishing cooperative outcomes are well-documented and have garnered significant scientific interest. For example, see \citet{adar2000free} and \citet{hughes2005free} for discussions of the rise of free-riding agents ultimately resulting in the decline of the peer-to-peer file sharing network Gnutella. This can be viewed as an instance of the Tragedy of the Commons~\cite{hardin1968tragedy}. Several research strands have consequently analyzed game-theoretic approaches to encourage cooperation in anonymous settings~\cite{johnny2010building,yang2012game}. We take a different perspective, focusing not necessarily on how to incentivize cooperation, but, rather, how to understand and compute the most cooperative stable outcome under the restrictions imposed by the game at hand.

Cooperation among near-anonymous agents interacting in pairs has also been studied in repeated games in the context of partner selection rules~\cite{zhang2016opting,rand2011dynamic,wang2012cooperation}. In prior work~\cite{anastassacos2020partner,leung2024learning,leung2024promote}, they find the emergent dominance of ``equivalent retaliation rules'' akin to Tit-for-Tat. The latter two works demonstrate that learning agents both learn and (as a majority) adopt the Out-for-Tat rule, in which players leave partners who deviate against them. This provides strong empirical support for our model, which assumes this behavior.

 Our starting point is a framework introduced by ~\citet{berker2024computing}, who study repeated symmetric games with restarts in which all agents follow an identical sequence of moves. They prove several fundamental results on equilibrium sequences in this restricted setting. 
\begin{theorem} \label{thm: emin basic results}
(Informal version of Proposition 1, Lemma 1, and Lemma 2 of \citet{berker2024computing}) In repeated symmetric games $\Gamma$ with restarts and discount factor $\beta$, where all agents follow an identical sequence of actions (\emph{i.e.}, the strategy is symmetric), we have each of the following:
\begin{enumerate}
    \item If there is some equilibrium sequence for $\Gamma$, then there exists an equilibrium sequence maximizing the agents' payoffs, which we call an \emph{optimal sequence.}
    \item Any optimal sequence will eventually reach a step in which both agents achieve a single payoff for the rest of the sequence. Call that payoff the \emph{goal value}. 
    \item For large enough $\beta$, the \emph{goal value} of any optimal sequence will be the highest  payoff of an action in $\Gamma$.
\end{enumerate}
\end{theorem}

Theorem \ref{thm: emin basic results} has a number of implications. It is natural to seek equilibria where the agents have the best cumulative outcomes. The first property says that such equilibria in fact exist (at least, for example, in any game with a pure Nash equilibrium), making it reasonable to study such equilibria. The second and third properties characterize the general structure of optimal equilbrium sequences for sufficiently large discount factor $\beta$. They begin with a \emph{hazing period}, in which the agents sacrifice utility to build mutual trust, followed by the agents reaping the reward of their camaraderie by receiving the goal value utility in each round thereafter.

Among such equilibrium sequences, some sequences require less hazing than others. Consider the game shown in Table \ref{symmetricgametable}. It is easy to check that both $(D,D),$ $(C_2, C_2),$ $(C_1, C_1)$, $(C_1, C_1), \ldots$ and  $(D,D),$ $(C_2, C_2),$ $(C_2, C_2),$ $\ldots,$ $(C_2, C_2),$ $(C_1, C_1)$, $(C_1, C_1)$, $\ldots$ are equilibrium sequences for sufficiently large $\beta$. However, the latter sequence delays the socially optimal action $C_1$ unnecessarily.

Therefore, \citet{berker2024computing} define an equivalence relation among sequences, yielding a more granular view of optimality and capturing the notion of optimality of a sequence also with respect to the amount of required hazing (see Section 4 of~\citet{berker2024computing}, limit-utility equivalence classes). This in turn motivates a natural computational problem: given a symmetric game, compute an optimal symmetric strategy sequence with minimal hazing. They show that this problem is (weakly) $\np$-hard, while also giving a pseudo-polynomial time algorithm.

In showing these results, the authors utilize a number of properties of optimal sequences in this restricted setting. For instance, they exploit the so-called ``threshold monotonicity'' (see Lemma 3 of \cite{berker2024computing}) property, which intutively states we can restrict our attention to sequences that order actions in terms of how much hazing they need before they can be played. This does not hold when extending to asymmetric strategies, as playing an action pair might require different amounts of hazing for each player.

In this work, we significantly relax the structural assumptions of~\cite{berker2024computing}, considering repeated games in which the players need not play the same action in each round and in which the game itself may be asymmetric. In the case of symmetric games and asymmetric strategies, this raises the question of which strategy each player will be assigned to when rematching. We primarily consider the case in which players follows the same strategy every time they rematch with a new partner. In Section~\ref{sec:randomreaasignment} we permit players to switch roles when rematching. Our main complexity and algorithmic results easily extend to this setting.

\subsection{Motivating Examples}
We begin by giving several examples to motivate our results, differentiate them from those in \cite{berker2024computing}, and highlight the complex behaviors that arise in this setting.

\paragraph{Symmetric games with asymmetric strategies.} Consider the game of two agents working on a series of group projects. Two distinct tasks must be done to complete each project, $T_1$ and $T_2$, and the agents only get utility $1$ if the project is complete. This symmetric game is represented in Table~\ref{tab: task game}. Note that, even though this is a symmetric game, no sequence of pairs of actions with both players always playing the same action will be stable: both players are incentivized to deviate when playing $(T_2,T_2)$ or playing $(T_1,T_1)$. Hence, no stable sequence exists in this game under the model considered in \cite{berker2024computing}. However, $(T_1, T_2)$ is a pure strategy Nash equilibrium, and, hence, repeating this action pair forever is a stable sequence in our model.

\begin{table}
\centering
\begin{minipage}{0.3\textwidth}
\centering
\begin{tabular}{| c | c |c| }
\hline
 {} & $T_1$ & $T_2$    \\
 \hline
 $T_1$ & $0, 0$ & $1, 1$ \\ 
 \hline
 $T_2$ & $1, 1$ & $0, 0$ \\
    \hline
 \end{tabular}
 
 \caption{Group Project}
\label{tab: task game}
\end{minipage}
\hspace{0.1\textwidth}
\begin{minipage}{0.45\textwidth}
\centering

\begin{tabular}{| c |c | c | c | c | }
\hline
 {} & $C$ & $D$ & $H_1$ & $H_2$ \\
 \hline
 $C$ & $99, 99$ & $0, 100$ & $0,0$ & $0,0$ \\ 
 \hline
 $D$ & $100,0$ & $0, 0$ & $0,0$ & $0,0$ \\
    \hline
 $H_1$ & $0,0$ & $0,0$  & $0,0$ & $5,50$ \\
 \hline 
 $H_2$ & $0,0$ & $0,0$ & $50,5$ & $0,0$ \\
 \hline
 \end{tabular}
 \caption{Nose Goes}
\label{tab: nose goes}
\end{minipage}
\end{table}

\paragraph{Some (even symmetric) games are unfair.}
Can we always distribute the hazing cost or utility fairly between agents? It turns out that sometimes agents must be hazed unequally to achieve minimum total hazing. 

Consider the game in Table~\ref{tab: nose goes}. The maximum social welfare outcome consists of $(C, C)$ repeating \emph{ad infinitum}, and, for large enough discount factor, it is possible to haze enough to disincentivize deviation from $(C, C)$ in only a single round. However, the minimum hazing sequence must include only one of either $(H_1, H_2)$ or $(H_2, H_1)$, leading to uneven total hazing. Inherently unfair games are perhaps less surprising in the asymmetric setting, but which games this holds for is not immediately obvious. We explore this question in Section~\ref{sec: formalism}. 

Due to the complexity of characterizing the ``fairness'' of stable sequences, we focus on formulating and solving corresponding optimization problems, the focus of Section~\ref{sec: algorithmic results}.

\subsection{Organization of the Paper}
In Section~\ref{sec: formalism}, we formalize the notions of (possibly asymmetric) equilibrium sequences in repeated games with restarts. In Section~\ref{sec:existence}, we characterize the conditions under which finite sequences of action pairs can form the ``goal sequences'' of stable sequences.  In Section~\ref{sec: algorithmic results} we consider this problem in the limit as the discount factor becomes negligible. In this regime, we define two optimization problems related to finding minimum hazing stable sequences and show that both are $\np$-hard. We also show that when the goal sequence is composed of maximum social welfare action pairs, there is a pseudo-polynomial time algorithm for solving the problem. Section~\ref{sec:randomreaasignment} addresses an alternative model where agents can change roles after restarting the game. In Section~\ref{sec: future}, we discuss several directions for future work. Appendix~\ref{appendix: stable} includes a discussion of when games have stable sequences, in a few special cases. Finally, all omitted proofs can be found in Appendix~\ref{appendix:omit}.

\section{Preliminaries} \label{sec: prelim}
Say $\Gamma$ is a two-player normal-form game, with a set of action pairs $A = A^{(1)} \times A^{(2)}$, where $A^{(i)} = \{a^{(i)}_1, a^{(i)}_2, \ldots , a^{(i)}_{n^{(i)}}\}$ is the set of actions available to player $i$.
\begin{itemize}
    \item We let $p^{(i)}: A \rightarrow \Z$ be the \emph{payoff} function of player $i$, taking as input a pair of actions of the two players and outputting an integer value. As shorthand, we also let $p^{(1)+(2)}$ denote $p^{(1)} + p^{(2)}$.
    \item Let $\beta \in (0,1)$ be the \emph{discount factor} such that if player $i$ receives payoff $p^{(i)}_t$ in timestep $t$, then her total discounted utility will be $\sum_{t = 0}^{\infty} \beta^t p^{(i)}_t$.
\end{itemize}

We will denote a game as a tuple, $\Gamma = (p^{(1)}, p^{(2)}, A)$. When $i$ refers to a player, we use $-i$ to refer to their opponent. 
The sequences of action pairs are $0$-indexed for consistency with the powers of the discount factor $\beta$. With $\N$ we denote the non-negative integers.

\section{Equilibrium Sequences} \label{sec: formalism}
We focus on strategies corresponding to sequences $\sigma = (\sigma_t^{(1)}, \sigma_t^{(2)})_{t \in \N}$, where $\sigma \in A^{\N}$ is a sequence of action pairs in $\Gamma$ such that player $i$ commits to playing $(\sigma_t^{(i)})_{t \in \N}$. Player $1$ will restart the sequence if player 2 deviates from $\sigma^{(2)}$ and vice versa. We also allow either player to restart the game after any round (even without deviating). This is a subtlety that does not arise in \cite{berker2024computing}. We illustrate the possibility of agents wanting to restart the game without deviating through the following example.
\begin{example}[Restarting without Deviating] \label{ex: re without de}
Consider the game in Table~\ref{tab: re without de}.
\begin{table}
\centering
\begin{minipage}{0.45\textwidth}
\centering

\begin{tabular}{| c | c |c| }
\hline
{} & $R$ & $C$ \\
 \hline 
 $r$ & $1, 0$ & $-100, -100$ \\ 
 \hline
 $c$ & $-100, -100$ & $0, 1$ \\
    \hline
 \end{tabular}
 
 \caption{Tightrope walking}
\label{tab: re without de}
\end{minipage}
\hspace{0.05\textwidth}
\begin{minipage}{0.45\textwidth}

\centering
\begin{tabular}{| c | c |c |}
\hline
 {} & $P$ & $S$  \\
 \hline
 $P$ & $1,0$ & $0, 1$ \\ 
 \hline
 $S$ & $0, 1$ & $0,1$ \\
    \hline
 \end{tabular}
 
 \caption{Doomed to suffer}
\label{tab: doom}
\end{minipage}
\end{table}

Consider any sequence (tightrope) $\sigma \in \{(r,R), (c,C)\}^{\N}$. Neither player can ever deviate or they receive $-100$ utility (fall off the tightrope). Hence, if players could only restart upon deviations of their opponents, any such $\sigma$ would be stable. However, if the players are permitted to restart after any round, one of the players can always ensure they receive utility $1$. Namely, if $\sigma_0 = (r,R)$, the row player can restart the game after the first round and similarly for the column player if $\sigma_0 = (c,C)$. We view the options for players to restart after any round as akin to typical assumptions of individual rationality.

\end{example}

Below we formally define the notion of a stable sequence of action pairs (\emph{i.e.}, a Nash equilibrium), in which no player can gain more utility by deviating or restarting.
\begin{definition}[Stable Sequences] \label{def: stable, fixed role}
We call a sequence $\sigma = (\sigma_t^{(1)}, \sigma_t^{(2)})_{t \in \N} \in A^\N$ \emph{stable} for discount factor $\beta$ if no player can increase their discounted utility by deviating or restarting the game at any timestep. Concretely, for player 1 we have, for all $k \in \N$ and  $a^{(1)} \in A^{(1)}$:
\begin{align*} 
    & \textstyle\sum_{t=0}^{k-1} \beta^t p^{(1)}(\sigma_t^{(1)}, \sigma_t^{(2)}) + \beta^k p^{(1)}(a^{(1)}, \sigma_k^{(2)}) 
    \\ &+ \textstyle \sum_{t=0}^{\infty} \beta^{k + 1 + t} p^{(1)}(\sigma_t^{(1)}, \sigma_t^{(2)}) \leq \sum_{t=0}^{\infty}\beta^t p^{(1)}(\sigma_t^{(1)}, \sigma_t^{(2)}), 
\end{align*}
The analogous inequalities must also hold for player 2.
\end{definition}
\begin{remark} \label{rmk: memoryless}
A careful reader might notice that Definition~\ref{def: stable, fixed role} only seems to consider players deviating a single time. This is because, if deviating only once cannot increase a player's utility, deviating more than once cannot either (this is an application of the ``one-shot deviation principle''). An analogous observation was made in \cite{berker2024computing}. 
\end{remark}

\begin{restatable}{proposition}{devonlyonce} \label{prop: deviate only once}
It benefits a player to deviate at least once if and only if it benefits a player to deviate once.  
\end{restatable}

The proof follows by a simple inductive argument, using that, post-deviation, the remaining game becomes a scaled version of the initial game. 

Notice that there can be infinitely many stable sequences for a given game (\emph{e.g.}, stable sequences of the form $(D,D),$ $(C_2, C_2),$ $(C_2, C_2),$  $\ldots,$ $(C_1,C_1),$ $(C_1, C_1),$  $\ldots$ in the game in Table~\ref{symmetricgametable}). Therefore, we would like to be able to (1) distinguish these sequences and (2) compute the most desirable among them. To do so, we formalize the notions of \emph{Pareto-optimality}, \emph{welfare maximization}, and \emph{limit-utility fairness}. In Section \ref{sec:existence}, we describe stable sequences in symmetric games with asymmetric strategies that satisfy all three properties.

\begin{remark}
In our study of stable sequences, we restrict our attention to sequences of the following form: a finite length prefix followed by an infinite periodic sequence of action pairs. We call the initial prefix the \emph{hazing period} and the finite sequence repeated infinitely thereafter the \emph{goal sequence}. A finite description length is a prerequisite for efficient computation, and general sequences need not necessarily admit one ---there is an uncountably infinite number of sequences but only a countably infinite number of finite descriptions. Moreover, periodicity allows us to take limits of the sums of differences of payoffs in sequences without concern for sequence convergence issues. This permits a natural way to compare the ``quality'' of sequences and formally define the related optimization problems of finding ``optimal'' stable sequences.
\end{remark}

\begin{definition}[Pareto-Optimal Sequence]
    Given a game $\Gamma = (p^{(1)}, p^{(2)}, A)$ and $\sigma, \tilde{\sigma} \in A^{\N}$, we say that  $\sigma$ \emph{surpasses} $\tilde{\sigma}$ if there exists $i \in [2]$ such that:
    \begin{align*}
        &\lim_{\beta \to 1} \textstyle\sum_{t=0}^{\infty} \beta^{t} \left( p^{(i)}(\sigma_t) - p^{(i)}(\tilde{\sigma}_t) \right) > 0, \; \text{while} 
         &\lim_{\beta \to 1} \textstyle\sum_{t=0}^{\infty} \beta^{t} \left( p^{(-i)}(\sigma_t) - p^{(-i)}(\tilde{\sigma}_t) \right) \geq 0.
    \end{align*}
    A sequence $\sigma \in A^{\N}$ is \emph{Pareto-optimal} (in $\beta \to 1$) if (1) it is stable for all sufficiently large $\beta$ and (2) it is not surpassed by any other stable sequence. 
\end{definition}

A notion stronger than Pareto optimality is welfare maximization, which we define for our context below.
\begin{definition}[Welfare maximization]
    Given a game $\Gamma = (p^{(1)}, p^{(2)}, A)$, a stable sequence $\sigma \in A^{\N}$  is \emph{welfare maximizing} (in the $\beta \to 1$ limit) if, for any other stable sequence
    $\tilde{\sigma} \in A^{\N}$, it holds that:
    \begin{align*}
        \lim_{\beta \to 1} \textstyle \sum_{t=0}^{\infty}\beta^{t}(p^{(1)+(2)}(\sigma_t) - p^{(1)+(2)}(\tilde{\sigma}_t)) \geq 0.
    \end{align*}
\end{definition}

We give a final desirable property of stable sequences.
\begin{definition}[Limit-utility fairness] \label{defn: lim utility fair}
    Given a game $\Gamma = (p^{(1)}, p^{(2)}, A)$, a stable sequence $\sigma \in A^{\N}$ is \emph{limit-utility fair} if there exists $T \in \N$ such that $\lim_{\beta \rightarrow 1} \sum_{t=T}^{\infty}\beta^{t} \left(p^{(1)}(\sigma_t) - p^{(2)}(\sigma_t) \right) = 0$.
\end{definition}
\begin{remark} \label{rmk: doom}
 Limit-utility fairness is not always possible in asymmetric games. For example, it is not possible in games where one player always receives strictly more utility than the other. Even when the players' utilities are normalized to be between 0 and 1, say by shifting their minimum utilities to each be $0$ and then scaling down, there are inherently ``unfair'' examples, such as the one shown in Table \ref{tab: doom}. In this game, if either player plays $S$, the row player ``suffers,'' receiving utility $0$. Indeed, the row player only receives utility $1$ if both players play $P$ (``seek and receive pity''). But, the sequence repeating $(S, S)$ forever is stable, and no stable sequence can ever include $(P,P)$ since the column player will always be incentivized to deviate to $(P, S)$.
\end{remark}

\section{Existence Results} \label{sec:existence}
Our starting point is the following:
\begin{center}
\emph{For which pairs $(\Gamma, \gamma)$, such that $\Gamma$ is a game and $\gamma \in A^r$,\\ can $\gamma$ be the goal sequence of a stable sequence of $\Gamma$?}   
\end{center}

Not all (even symmetric) games admit stable sequences, \emph{e.g.},  Rock-Paper-Scissors. Moreover, although all $2 \times 2$ symmetric games have a stable sequence, this does not hold for asymmetric games. We discuss these nuances in Appendix~\ref{appendix: stable}.

Theorem~\ref{thm: folk theorem equil} characterizes which goal sequences can arise in stable sequences. For convenience, we define the goal value of a goal sequence $\gamma$ as the average per-round payoffs obtained in the goal sequence when $\beta \to 1$. We also introduce notation for the deviation payoffs for an action pair.

\begin{definition}[Goal value]
Given a game $\Gamma = (p^{(1)}, p^{(2)}, A)$ and a goal sequence $\gamma \in A^{r}$, its corresponding goal value is 
\begin{equation*}
 v_{\gamma} := 
 (v^{(1)}_\gamma, v^{(2)}_\gamma) = \textstyle \left(\frac{1}{r}\sum_{j = 1}^r p^{(1)}(\gamma_j), \frac{1}{r}\sum_{j = 1}^r p^{(2)}(\gamma_j)\right).   
\end{equation*}
\end{definition}

\begin{definition}[Deviation payoff]
  Given a game $\Gamma = (p^{(1)}, p^{(2)}, A)$, and action pair $a=(a^{(1)},a^{(2)})\in A$, define
\begin{equation*}
    d_a^{(1)} := \max_{\tilde{a}^{(1)}} p^{(1)}(\tilde{a}^{(1)}, a^{(2)}), \quad d_a^{(2)} := \max_{\tilde{a}^{(2)}} p^{(2)}(a^{(1)}, \tilde{a}^{(2)}).
\end{equation*}
\end{definition}

\begin{restatable}{theorem}{stableexistence}
\label{thm: folk theorem equil}
Let $\Gamma = (p^{(1)}, p^{(2)}, A)$ be a game and $\gamma \in A^{r}$.
\begin{enumerate}
    \item Suppose there exists $a \in A$  such that $ d_a^{(1)} < v_{\gamma}^{(1)}$
and $d_a^{(2)} < v_{\gamma}^{(2)}$. Then, for large enough $\beta \in (0,1) $ and  $T \in \N$, the sequence $\sigma$ repeating $a$ for $T$ time steps and then repeating $\gamma$ forever is stable. 
    \item If for all $a \in A$ we have $ d_a^{(1)} > v_{\gamma}^{(1)}$ or $ d_a^{(2)} > v_{\gamma}^{(2)}$ then, for any $\beta$, no stable sequence has $\gamma$ as its goal sequence. 
\end{enumerate}
\end{restatable}

The proof of the first part is inspired by the Folk theorem. The second part follows from considering stability at the first action pair in a candidate stable sequence.

\begin{restatable}{corollary}{stableexistencecor} \label{cor: folk theorem equil}
Let $\Gamma$ be symmetric and suppose there exists $a_* = (a^{(1)}_*, a^{(2)}_*), a \in A$ with $a_*$ maximum social welfare and $d_a^{(1)},  d_a^{(2)} < (p^{(1)}(a_*) + p^{(2)}(a_*))/2$. Then $\gamma = ( (a_*^{(1)}, a_*^{(2)}), (a_*^{(2)}, a_*^{(1)}))$ is the goal sequence of some stable sequence in $\Gamma$ by Theorem~\ref{thm: folk theorem equil}. Moreover, this stable sequence is Pareto-optimal, limit-utility fair, and welfare-maximizing.
\end{restatable}

\section{Computing Minimum Hazing Sequences} \label{sec: algorithmic results}

We next consider the problem of computing stable sequences with minimum hazing in the $\beta \to 1$ limit. We begin by defining the \emph{hazing cost} and \emph{threshold} for a given action pair.

\begin{definition}[Hazing Cost, Threshold] \label{defn: hazing, thresh}
    For a game $\Gamma = (p^{(1)}, p^{(2)}, A)$, goal sequence $\gamma \in A^r$, $a \in A$, and $i \in [2]$, we define the \emph{hazing cost} $h_{a}^{(i)} := v_{\gamma}^{(i)} - p^{(i)}(a)$ and the \emph{threshold} $t_{a}^{(i)} := d_a^{(i)} - v_{\gamma}^{(i)}$ for player $i$.
\end{definition}
The \emph{hazing cost} of an action pair for a player defines how much utility that player loses in the long run by playing that action compared to her average utility in the goal sequence.\footnote{Note that the hazing cost could be negative for some actions.} Intuitively, we want the hazing sequence to have a sufficiently high cost for both players to disincentivize them from taking an action that would result in restarting the sequence. As we will see, the \emph{threshold} of an action pair for a given player defines the amount of total hazing that that player must have accumulated before playing that action in order to guarantee that they will not deviate. In Theorem~\ref{thm: stabilitywithlimit}, we will make use of these two definitions to give a sufficient and necessary condition for stability in the $\beta \to 1$ limit.

For conciseness, we also define notation for total hazing and the threshold for a goal sequence.

\begin{definition}[Total Hazing]
Let $\sigma \in A^{\N}$ be a stable sequence with goal sequence $\gamma \in A^r$ for a game $\Gamma = (p^{(1)}, p^{(2)}, A)$. For each $k \in \N$, define the total hazing up to time $k$, $H_k := (H^{(1)}_k, H^{(2)}_k) = \sum_{t = 0}^k h_{\sigma_t}$,
to be the sum of hazing costs of the first $k + 1$ actions in the hazing period.
\end{definition}

\begin{definition}[Threshold for a goal sequence] \label{defn: thresh goal seq}
For a goal sequence $\gamma \in A^{r}$ and $i \in [2]$, define its \emph{threshold} as:
\begin{align}
\textstyle\theta_{\gamma}^{(i)} = \max_{k \in [r]} \left( t_{\gamma_k}^{(i)}  - \sum_{t = 1}^{k -1} h_{\gamma_t}^{(i)}\right). \label{seq_threshold}
\end{align}
\end{definition}

In words, for each $k \in [r]$, we need to surpass the threshold for $\gamma_k$ upon reaching it, which is only  possible by accumulating $\theta_\gamma^{(i)}$ hazing for each player $i$ by the time the goal sequence is reached. The summation in equation (\ref{seq_threshold}) accounts for the change in total hazing (after the goal sequence begins) up to the $(k-1)^{\text{th}}$ action pair in the goal sequence. 

Finally, Theorem~\ref{thm: stabilitywithlimit} defines the stability of a sequence in the  $\beta \to 1$ limit. 

\begin{restatable}{theorem}{stabilitywithlimit} \label{thm: stabilitywithlimit}
Let $\Gamma = (p^{(1)}, p^{(2)}, A)$ be a game. A sequence $\sigma \in A^{\N}$ with finite hazing period and goal sequence $\gamma \in A^r$ is stable for all sufficiently large $\beta \in (0,1)$ if and only if, for all $k \in \N$ and $i \in [2]$, we have  $H_{k-1}^{(i)}  >  t_{\sigma_k}^{(i)}.$
\end{restatable}

The intuition for the strict inequality here is that ties break in favor of deviating, since the deviation payoff comes earlier.

We can also characterize stability in the limit $\beta \to 1$ using the language of thresholds of the goal sequence.
\begin{corollary} \label{corollary: goal sequence thresh stab}
A sequence $\sigma$ formed by a hazing period of length $T$ and a repeated goal sequence $\gamma \in A^r$ is stable in the $\beta \to 1$ limit if and only if $\sigma$ is stable in the hazing period in the $\beta \to 1$ limit and, for $i \in [2]$, $H_T^{(i)} > \theta_{\gamma}^{(i)}$.
\end{corollary}

Using the necessary and sufficient conditions for stability from Corollary~\ref{corollary: goal sequence thresh stab}, we now define the computational problem of finding sequences inducing the minimum possible hazing.

\begin{definition}[$\prob$] \label{defn: genoptrep}
Denote the hazing and threshold tuples for each action pair and each player in a game as $\Bigl\{ \left(h^{(1)}_{a}, h^{(2)}_{a}, t^{(1)}_{a}, t^{(2)}_{a} \right) \Bigl\}_{a \in A} \in \left(\left(1/r \cdot \Z\right)^{4}\right)^{|A|}$
    . Also let $(\theta^{(1)}, \theta^{(2)}) \in (\frac{1}{r} \cdot \Z)^{2}$ be the thresholds for a goal sequence for each player. Given $\Delta > 0$, $\prob$ asks to find a sequence $\sigma \in A^{\ell}$ (for any finite $\ell$) such that the sum of the total hazings satisfies $H_{\ell-1}^{(1) + (2)} = \sum_{t=0}^{\ell - 1} h_{\sigma_t}^{(1) + (2)} \leq \Delta$, subject to:
\begin{enumerate}
    \item $H_{\ell-1}^{(i)}  > \theta^{(i)}$, $\quad \forall i \in [2]$
    \item $ H_{k-1}^{(i)} > t_{\sigma_k}^{(i)}, \quad \forall k \in \{0, \cdots, \ell-1\}, \quad \forall i \in [2]$
\end{enumerate}
\end{definition}

Notice that, by Theorem \ref{thm: folk theorem equil}, in some games it is easy to compute hazing sequences that induce stable sequences with goal sequence $\gamma$. In fact, if there exists some action pair $a \in A$ that satisfies the conditions of Theorem \ref{thm: folk theorem equil}, repeating $a$ sufficiently many times makes for such a hazing sequence. Moreover, the sum of the players' total hazing will be at most: 
\[
\bd := 
\theta^{(1)}_{\gamma} + \theta^{(2)}_{\gamma} + h_a^{(1)} + h_a^{(2)} \leq (r +2)\kappa,
\]
where $\kappa$ is the difference between the largest and smallest possible payoff values in $\Gamma$. However, this trivial upper bound, $\bd$, could be arbitrarily larger than the minimum hazing possible, which is what we explore in this section.

\begin{remark}\label{remark:rationalhazings}
In the instance of $\prob$ induced by $\Gamma$ and a finite length goal sequence $\gamma \in A^r$, we have $(\theta_{\gamma}^{(1)}, \theta_{\gamma}^{(2)}) \in (\frac{1}{r} \cdot \Z)^2$ and $\{(h^{(1)}_{\sigma}, h^{(2)}_{\sigma}, t^{(1)}_{\sigma}, t^{(2)}_{\sigma})\}_{\sigma \in A} \in ((\frac{1}{r} \cdot \Z)^{4})^{|A|}$. This follows directly from Definitions~\ref{defn: hazing, thresh} and~\ref{defn: thresh goal seq}. We make heavy use of this fact in Algorithm~\ref{alg:dp}.
\end{remark}

\begin{restatable}{theorem}{nphardness} \label{thm: np hardness}
$\prob$ is (weakly) $\np$-hard. 
\end{restatable}
\emph{Proof Sketch.} The key idea is to reduce from the Unbounded Subset-Sum Problem with non-negative integers. We define a symmetric game  in which: all but the threshold for the goal sequence action pair are trivially met and only the main diagonal action pairs are viable in a minimum hazing stable sequence. The payoffs on the main diagonal can then be chosen such that their hazing costs correspond to the integers from the instance of Unbounded Subset-Sum.

Although $\prob$ is well-defined in broad generality, we are mostly interested in the problem of computing minimum hazing sequences for welfare-maximizing stable sequences with infinitely repeated finite-length goal sequences. So, we define the following computational problem. 

\begin{definition}[$\welfprob$]
Consider a game, $\Gamma = (p^{(1)}, p^{(2)}, A)$, and a goal sequence, $\gamma \in A^r$, where for each $t \in [r]$, $\gamma_t \in A$ is a maximum social welfare action pair. Suppose also that $\gamma$ is the goal sequence of a stable sequence, $\sigma$, that achieves total sum of hazings $\bd$. Then $\welfprob(\Gamma, \gamma, B)$ is the instance of $\prob$ induced by $\Gamma$, $\gamma$, and total hazing bound $B$. 
\end{definition}

\begin{remark} \label{remark:sumofhazingsperround}
Since $\gamma_t \in A$ is maximum social welfare for each $t \in [r]$, we know that each $a \in A$ has $h^{(1)}_a + h^{(2)}_a \geq 0$. Indeed, if not, then $a$ would induce higher social welfare than the per-round average payoff in $\gamma$, a contradiction.
\end{remark}

We will use the structure of $\welfprob$ to show that for any stable sequence with goal sequence $\gamma \in A^r$, there exists another highly structured stable sequence with goal sequence $\gamma$ and minimum hazing (Lemma~\ref{lem: hazing per struct}). This insight will allow us to solve  $\welfprob$ in pseudo-polynomial time (Theorem \ref{thm:dp}).

\begin{restatable}{lemma}{hazingstruct} \label{lem: hazing per struct}
For a game $\Gamma$, let $\gamma \in A^r$ be a maximum social welfare goal sequence, such that there exists a stable sequence with goal sequence $\gamma$ and total hazing $\bd$. Then, there exists a \emph{minimum hazing} stable sequence $\sigma$ with goal sequence $\gamma$, such that for $k$ and $k_1 < k_2$ in the hazing period:
\\ \textbf{\textup{1. Total hazing bound:}} $0 \leq H_k^{(1)}, H_k^{(2)}, H_k^{(1) + (2)} \leq \bd.$
\\ \textbf{\textup{2. Monotonicity of total hazing:}} $H_{k_2}^{(1) + (2)}  - H_{k_1}^{(1) + (2)}\geq 0.$
\\ \textbf{\textup{3. Injectivity of total hazing:}} $(H_{k_1}^{(1)}, H_{k_1}^{(2)}) \neq (H_{k_2}^{(1)}, H_{k_2}^{(2)}).$

\end{restatable}
\emph{Proof Sketch.} The first two properties follow from stability and the fact that each $ \gamma_t \in \gamma$ is maximum social welfare. The third property exploits the observation that we can remove segments of the hazing period between repeated pairs of total hazing values without compromising stability.

Although $\prob$ is not obviously in $\np$ in general (\emph{e.g.}, optimal stable sequences can have exponentially long hazing periods), $\welfprob$ is indeed in $\np$ under some weak additional structural assumptions.

\begin{restatable}{theorem}{npcompleteness}\label{thm: np-completeness}
$\welfprob(\Gamma, \gamma, B)$ (in its decision version) is in $\np$ for classes of games $\Gamma$, goal sequence $\gamma \in A^r$ with $|A| = n$ and $r = \poly(n)$, and $\bd$ such that either one of the following conditions is satisfied:
\begin{enumerate}
    \item There is a stable sequence with goal sequence $\gamma$ with total hazing $\poly(n)$, \emph{e.g.}, $B = \poly(n)$. 
    \item There exists a stable sequence with at most $\poly(n)$ action pairs with negative hazing value for either player.
\end{enumerate}
\end{restatable}
\emph{Proof Sketch.} The first sufficient condition follows from the second and third properties of Lemma~\ref{lem: hazing per struct} (or Theorem~\ref{thm:dp}). The second condition follows from the fact that, as long as the threshold for some $a \in A$ is met at some time step $t$ and it does not contribute negative hazing to either player, it can be inserted at time step $t$ without disrupting stability. This allows for clustering all such $a \in A$ into consecutive runs, yielding hazing sequences with succinct representations.

\begin{remark}
The first part of Theorem~\ref{thm: np-completeness} holds when the utilities in $\Gamma$ are bounded by $\poly(n)$ and there exists  $a \in A$ such that $v^{(i)}_\gamma > d_a^{(i)}$ for $i \in [2]$. Indeed, by Theorem~\ref{thm: folk theorem equil}, since the threshold for $\gamma$ is $\poly(n)$, repeating $a \in A$ for $\poly(n)$ iterations, followed by cycling through $\gamma$, results in a stable sequence with goal sequence $\gamma$ and $\poly(n)$ total hazing.
\end{remark}

Finally, we give a dynamic programming algorithm for solving $\welfprob$ (Algorithm~\ref{alg:dp}) and prove that it runs in pseudo-polynomial time (Theorem \ref{thm:dp}).

We start by giving an overview of Algorithm \ref{alg:dp}. Given as input the tuples of thresholds and hazing costs for each action pair, the threshold for the goal sequence, and the upper bound $\bd$ on the total hazing, the algorithm will first construct an empty queue, $Q$, and a matrix, $\mem$, indexed by all possible pairs of hazing values (rationals with denominator $r$ between $0$ and $\bd$ from Remark~\ref{remark:rationalhazings}). The algorithm will start with a total hazing value of 0 for each player (\emph{i.e.}, starts at the upper left entry of $\mem$), and, from that pair of hazing values, it will consider all possible action pairs available to the agents. For each action pair that satisfies the conditions of Lemma \ref{lem: hazing per struct} and whose threshold is met, the algorithm will compute the new reachable total hazing for each player by adding the hazing cost of that action pair to the the total hazing of each player so far. That will ``move'' the algorithm to a different entry in $\mem$, and the algorithm will enqueue that entry to be explored later. The algorithm then proceeds in a breadth-first-search manner, dequeuing pairs of hazing values from $Q$ (\emph{i.e.}, ``visiting'' entries in $\mem$) and considering all other entries in $\mem$ that can be reached by taking ``valid'' action pairs. While visiting each entry in $\mem$, the algorithm keeps track of the minimum total hazing value pair encountered so far that satisfies the goal threshold, as well as the appropriate information to reconstruct the minimum hazing sequence. The algorithm terminates when $Q$ is empty. Recall that we use $h^{(1) + (2)}$ as shorthand denoting $h^{(1)} + h^{(2)}$ (see Section~\ref{sec: prelim}).

\begin{algorithm}[t!]
    \caption{Dynamic Programming Algorithm for $\welfprob$}\label{alg:dp}
    \begin{algorithmic}[1]
    \STATE \textbf{Input:} (1.) Goal sequence threshold $\theta := (\theta^{(1)}, \theta^{(2)}) \in (\frac{1}{r} \cdot \Z)^2$, (2.) action pair hazing costs and threshold $\{(h^{(1)}_{a}, h^{(2)}_{a}, t^{(1)}_{a}, t^{(2)}_{a})\}_{a \in A} \in ((\frac{1}{r} \cdot \Z)^{4})^{|A|}$, and (3.) initial total hazing bound $\bd$.
    \STATE \textbf{Output:} Hazing sequence, $\sigma_* \in A^{\ell}$, and total hazing $H_*^{(1)}+ H_*^{(2)} \in \frac{1}{r} \cdot \Z$.
    \STATE Create an empty queue, $Q$
    
    \STATE  $\mem \leftarrow [] \times []$ \COMMENT{Indexed by pairs in $[0, \bd] \times [0, \bd]$}
    \STATE $(H_*^{(1)}, H_*^{(2)}) \leftarrow (\bd, \bd)$  \COMMENT{Min.\ hazing above $\theta$ so far}
    \STATE Enqueue$(0,0, \none, \none)$ \COMMENT{Hazing pairs to process}
    \WHILE{$Q$ is not empty} \label{while}
        \STATE $(H^{(1)}, H^{(2)}, p, a_p) \leftarrow$ Dequeue($Q$) \COMMENT{$H^{(i)}$ is current hazing for player $i$}
        \STATE $\mem[H^{(1)}, H^{(2)}] = (p, a_p)$  \COMMENT{Parent (hazing, action)}
        \IF{ $ H^{(i)} > \theta^{(i)}$ $\forall i \in [2]$ \textbf{and} $ H^{(1) + (2)} < H^{(1) + (2)}_*$} \label{opt_threshold}
            \STATE $(H_*^{(1)}, H_*^{(2)})  \leftarrow  (H^{(1)}, H^{(2)})$ \COMMENT{Update optimal} 
            \STATE \textbf{continue}
        \ENDIF
        \FOR{$a := (a^{(1)}, a^{(2)}) \in A$} \label{step:innerfor}
            \IF{$\exists i \in [2]$ s.t. $H^{(i)} < t_{a}^{(i)}$ \textbf{or} $(H^{(1)} + h_a^{(1)}, H^{(2)} + h_a^{(2)}) \in \mem$ } \label{step:revisited}
                    \STATE \textbf{continue} \COMMENT{Thresh.\ unsatisfied or redundant}
            \ENDIF
            \IF{$H^{(1) + (2)} + h_a^{(1) + (2)} > B$ \textbf{or} $\exists i \in [2]$  s.t. $H^{(i)} + h_a^{(i)} < 0$} \label{step:outofbounds}
                \STATE \textbf{continue} \COMMENT{Invalid or irrelevant hazing pair}
            \ENDIF
            \STATE Enqueue$((H^{(1)} + h_a^{(1)},H^{(2)} + h_a^{(2)}, (H^{(1)}, H^{(2)}), a)$
        \ENDFOR
    \ENDWHILE  
    \STATE $\sigma_* \leftarrow ()$
    \STATE $(H^{(1)}, H^{(2)}) \leftarrow (H_*^{(1)}, H_*^{(2)})$ 
    \WHILE{$\mem[H^{(1)}, H^{(2)}][2] \neq \none$}\label{step:reconstruct}  
        \STATE $\sigma_{*}\text{.prepend}(\mem[H^{(1)}, H^{(2)}][3])$ \COMMENT{Rebuild $\sigma_{*}$}
        \STATE $(H^{(1)}, H^{(2)}) \leftarrow \mem [H^{(1)}, H^{(2)}][2]$
    \ENDWHILE
    
    \STATE \textbf{return:}  $\sigma_*, H_*^{(1) + (2)}$ \label{return}
\end{algorithmic}
\end{algorithm}

\begin{restatable}{theorem}{thmdp} \label{thm:dp}
$\welfprob(\Gamma, \gamma, B)$ is solvable in $O(\poly(n)r^2B^2)$ time and $O(r^2B^2)$ space, where $|A| = n$ and $\gamma \in A^r$.
\end{restatable} 
\emph{Proof Sketch.} We use the structure of the min-hazing sequence shown in Lemma \ref{lem: hazing per struct} to upper bound the size of $\mem$ that needs to be searched to find that sequence. This allows us to bound the time and space complexity of the algorithm.

Note that while Algorithm~\ref{alg:dp} minimizes $H^{(1) + (2)}$ as written, it  can be modified to minimize any function of the hazing costs $H^{(1)},  H^{(2)}$  by modifying Steps~\ref{opt_threshold} and~\ref{return} accordingly.

\section{Random Player Reassignment} \label{sec:randomreaasignment}
In this section, we consider a variant of our framework, where the roles of players may switch after a rematch. Indeed, in other settings, it has been shown that such role-switching can promote cooperation~\cite{moon2016role}. As before, we assume each player is required to play at least one round of the game with their partner, and can choose to leave after any round. When reassigned to a new relationship, they take on each of the two possible player roles with probability $1/2$. We introduce the analogous notion of stability below.

\begin{definition}[Stable Sequences with Random Player Reassignment] \label{def: stable, reassign}
We call a sequence $\sigma = (\sigma_t^{(1)}, \sigma_t^{(2)})_{t \in \N} \in A^{\N}$ stable if no player can increase their expected discounted utility by deviating at any timestep. This means for all $k \in \N$ and $i \in [2]$ we have
$\sum_{t=0}^{k-1} \beta^t p^{(i)}(\sigma_t) + \beta^k d_{\sigma_k}^{(i)} +  \sum_{t=0}^{\infty} \beta^{k+1 + t} \frac{p^{(i)}(\sigma_t) + p^{(-i)}(\sigma_t)}{2} \leq \sum_{t=0}^{\infty}\beta^t p^{(i)}(\sigma_t)$.

\end{definition}
With Definition \ref{def: stable, fixed role}, defecting multiple times was profitable for a player only if defecting once was profitable for that player as well. With Definition~\ref{def: stable, reassign}, on the other hand, it can be shown that defecting multiple times is profitable for a player only if defecting once is profitable \emph{for some player}, which follows from an argument similar to  Proposition~\ref{prop: deviate only once}. Ensuring no player can profit from deviating once is therefore still sufficient to ensure overall stability.

Just as before, instances of  $\welfprob$ and $\prob$  are induced by this variant of the repeated games framework (each action pair has a threshold value, hazing value, etc.). It is not hard to show that the same $\np$-hardness result holds (by the same construction as in Theorem~\ref{thm: np hardness}), and Algorithm~\ref{alg:dp} still solves $\welfprob$ in pseudo-polynomial time.

\section{Directions for Future Work} \label{sec: future}

There are several interesting directions for future work. One avenue is to permit mixed strategies. How would agents verify whether their opponent adhered to their assigned strategy or deviated to another strategy with the same support? Cryptographic protocols, such as those in~\citet{blum1983coin}, are a good candidate approach for strategy verification. 

Another direction is to extend our algorithmic results beyond maximum social welfare goal sequences.  The main difficulty in this general setting is that net ``unhazing'' is now possible, so total hazing is no longer monotonic. Our dynamic programming algorithm can be tweaked to work for arbitrary goal sequences with the additional constraint that no player's hazing can ever exceed some fixed threshold $\bd$. But, in theory, the optimal stable sequence could require hazing both players an enormous amount before  unhazing can be applied to reduce the total hazing down to the optimal amount. 

A third direction is to extend our work to fixed discount factors $\beta$. While some of our results hold for large enough $\beta$ (e.g., Theorem~\ref{thm: folk theorem equil} and Corollary~\ref{cor: folk theorem equil}), our algorithmic approach does not; for fixed $\beta$, the thresholds for actions depend on the time they are played.  

Our results extend to games with any number of players, although Algorithm~\ref{alg:dp}'s runtime scales exponentially.

\section*{Acknowledgements}
   We are grateful to Tuomas Sandholm and Brian Hu Zhang for their valuable feedback on an early manuscript and  George Z. Li for helpful discussions in the early stages of this project. We also thank our IJCAI'25 reviewers for their generous and detailed reviews.

   This material is based upon work supported by the National Science Foundation Graduate Research Fellowship Program under Grant No.\ DGE2140739. Any opinions, findings,
and conclusions or recommendations expressed in this material are those of the author(s)
and do not necessarily reflect the views of the National Science Foundation.

R.E.B. thanks the Cooperative AI Foundation, Macroscopic Ventures (formerly Polaris Ventures / the Center for Emerging Risk Research) and Jaan Tallinn's donor-advised fund at Founders Pledge for financial support. R.E.B. is also supported by the Cooperative AI PhD Fellowship.

\bibliography{references}

\appendix

 \section{Existence of Stable sequences} \label{appendix: stable}
 \begin{example}[Goal sequences need not exist, even in symmetric games]
Consider the classic game rock-paper-scissors, represented in Table~\ref{tab: rps}, where players only receive positive utility if they win.
\begin{table}[b]
\centering
\begin{minipage}{0.45\textwidth}
\centering

\begin{tabular}{| c | c |c | c |}
\hline
 {} & $R$ & $P$ & $S$ \\
 \hline
 $R$ & $0,0$ & $0, 1$ & $1,0$  \\ 
 \hline
 $P$ & $1, 0$ & $0, 0$ &  $0,1$ \\
    \hline
 $S$   & $0,1$ & $1,0$ & $0,0$ \\
 \hline
 \end{tabular}
 
 \caption{Rock-paper-scissors}
\label{tab: rps}
\end{minipage}
\hspace{0.05\textwidth}
\begin{minipage}{0.3\textwidth}

\centering
\begin{tabular}{| c | c |c| }
\hline
{} & $H$ & $T$ \\
 \hline 
 $H$ & $0, 1$ & $1, 0$ \\ 
 \hline
 $T$ & $1, 0$ & $0, 1$ \\
    \hline
 \end{tabular}
 
 \caption{Matching Pennies}
\label{tab: stalemate}
\end{minipage}
\end{table}

No stable sequence can offer both agents a per-round average utility of $1$. Then, for the agent who does not receive a per-round average utility of $1$, they are incentivized to deviate on the first action in which they do not receive utility $1$ (then receiving utility $1$). Hence, rock-paper-scissors has no stable sequences.    
\end{example}

As a counterpoint, we note that all $2 \times 2$ symmetric games have stable sequences. 
\begin{proposition} \label{prop: stable seq 2by2}
All $2 \times 2$ symmetric games have stable sequences.
\end{proposition}
\begin{proof}
This is an immediate corollary of the following two facts:
\begin{enumerate}
    \item All $2 \times 2$ symmetric games have a pure strategy Nash equilibrium (this is follows from simple casework).
    \item If $(C_1, C_2)$ is a Nash equilibrium of a game $\Gamma$, the sequence repeating $(C_1, C_2)$ forever is a stable sequence in $\Gamma$.
\end{enumerate}
\end{proof}

However, Proposition~\ref{prop: stable seq 2by2} does not hold in $2 \times 2$ asymmetric games.
\begin{example}[$2 \times 2$ asymmetric game without stable sequences]
Consider the classic game Matching Pennies (see Table~\ref{tab: stalemate}). This game has no pure-strategy Nash equilibria. It is easy to check that it also does not have any stable sequences by a similar argument to rock-paper-scissors in Table~\ref{tab: rps}. 

\end{example}

\section{Omitted proofs} \label{appendix:omit}
In this section, we restate key claims whose formal proofs were omitted in the main body of the text and provide their proofs.
\devonlyonce*
\begin{proof}
The reverse direction is trivial. We prove the forward direction by induction. Suppose it does not benefit a player to defect at most $r$ times, for $r \geq 1$.  We want to show that it cannot benefit a player to defect exactly $r +1$ times. Indeed, after defecting once, the remaining game is a scaled version of the original game in which the player defects $r$ times. So, by the inductive hypothesis, their total utility after defecting is at most as much as if they did not defect again after the first defection. But then their total utility is at most as much as if they only defected once. But, since it does not benefit a player to defect once (our base case), their utility from defecting $r + 1$ times is then at most their utility if they did not defect at all. The claim then follows by induction.
\end{proof}
\stableexistence*
\begin{proof}
\emph{(Claim 1.)} The  proof is inspired by the classic Folk theorem from repeated games.  
Let
\[
\delta := \min \left(v_\gamma^{(1)} - d_a^{(1)}, v_\gamma^{(2)} - d_a^{(2)}\right),
\]
and let $T$ be a positive integer large enough so that $T \cdot \delta > r\kappa$, where $\kappa$ is largest payoff value achievable in $\Gamma$. 
We appeal to Theorem~\ref{thm: stabilitywithlimit}. Rephrasing the guarantee from Theorem~\ref{thm: stabilitywithlimit}, we get that $\sigma$ is stable for sufficiently large $\beta \in (0,1)$ if and only if, for all $k \in \N$ and $i \in [2]$, we have
\begin{align}
  kv_{\gamma}^{(i)} - \sum_{t = 0}^{k-1} p^{(i)}(\sigma_t) &> d^{(i)}_{\sigma_k} - v_{\gamma}^{(i)} \label{eq: restart cost vs. profit} \\ \iff
(k+1)v_{\gamma}^{(i)}  &> d^{(i)}_{\sigma_k} + \sum_{t = 0}^{k-1} p^{(i)}(\sigma_t) \label{eq: key ineq}
\end{align}
Equation (\ref{eq: restart cost vs. profit}) intuitively says that the long-term cost of restarting a relationship is larger than the immediate profit from defecting. Now, if $k$ is during the hazing period (\emph{i.e.}, $k \leq T - 1$), then $\sigma_t = a$ for $t \leq k$, and (\ref{eq: key ineq}) holds because $p^{(i)}(a) \leq d_{a}^{(i)} < v_\gamma^{(i)}$. Otherwise, if $k > T-1$, we have 

\begin{align*}
\sum_{t = 0}^{k-1} p^{(i)}(\sigma_t) &\leq T d_{a}^{(i)} + \sum_{t = T}^{k-1}p^{(i)}(\sigma_t) \\
&\leq T d_{a}^{(i)} + \left\lfloor \frac{k - T }{r} \right\rfloor \cdot r v_\gamma^{(i)} + \left( (k - T) - \left\lfloor \frac{k - T}{r} \right\rfloor \cdot r \right) \kappa \\
&\leq T d_{a}^{(i)} + (k - T - r + 1)v_\gamma^{(i)} + (r -1) \kappa.
\end{align*}
In particular, then we have
\begin{align*}
(k+1)v_{\gamma}^{(i)} - d^{(i)}_{\sigma_k} - \sum_{t = 0}^{k-1} p^{(i)}(\sigma_t) &\geq (T + r)v_\gamma^{(i)} - r \kappa - Td_a^{(i)} \\
&> T\delta -r \kappa  > 0,
\end{align*}
yielding the desired result.

\emph{(Claim 2.)} Suppose that $\sigma$ were a stable sequence with $\gamma$ as its goal sequence. Then, since $\sigma_0 \in A$, for some $i \in [2]$ we have $d^{(i)}_{\sigma_0} > v_\gamma^{(i)}$. Suppose without loss of generality this holds for $i = 1$. Then, player 1 can repeatedly deviate on the first round of the sequence and gain strictly higher utility than what she would have gotten upon reaching the goal sequence, contradicting stability. Since the player gets higher utility than the goal sequence \emph{in every timestep}, the result holds for any value of $\beta$.
\end{proof}

\stableexistencecor*
\begin{proof}
The proof of that there exists $\sigma$ stable with goal sequence $\gamma$ for large enough $\beta$ follows immediately from Theorem~\ref{thm: stabilitywithlimit} and the observation that, for each $i \in [2]$, $v_{\gamma}^{(i)} = (p^{(1)}(a_*) + p^{(2)}(a_*))/2$. Then, the welfare maximization (and hence Pareto optimality) follows from the fact that $a_*$ is maximum social welfare. Finally, $\sigma$ is limit-utility fair because $v_\gamma^{(1)} = v_\gamma^{(2)}$.
\end{proof}

\stabilitywithlimit*

\begin{proof}
 We know that $\sigma$ is stable for $\beta > 0$, if and only if, for player $1$, for all $k \in \N$, we have  
 \begin{align} 
    \sum_{t = 0}^{k-1} \beta^{t}p^{(1)}(\sigma_t) + \beta^k \cdot \max_{a^{(1)} \in A^{(1)}} p^{(1)}(a^{(1)}, \sigma_k^{(2)})  \leq  (1 - \beta^{k+1})\sum_{t = 0}^{\infty} \beta^t p^{(1)}(\sigma_t), \label{eq: stability with repeats}
 \end{align}
 and similarly for player $2$.  Suppose without loss of generality that the hazing period of $\sigma$ is length $T \geq r$ (treat the first iteration of $\gamma$ as part of the hazing period if necessary). Then, the right hand side 
 of (\ref{eq: stability with repeats}) can be expressed as:
 \begin{equation*}
     (1 - \beta^{k+1}) \cdot \left(\sum_{t = 0}^{T-1} \beta^t p^{(1)}(\sigma_t) + \frac{\beta^{T}}{1 - \beta^r}\left( \sum_{t = 0}^{r-1} \beta^t p^{(1)}(\gamma_t)\right)\right).
 \end{equation*}

 So, we equivalently require:
 \begin{align}\label{eqtn: rearranged stab with repeats}
  &\sum_{t = 0}^{k-1} \beta^{t}p^{(1)}(\sigma_t) + \beta^k \max_{a^{(1)}} p^{(1)}(a^{(1)}, \sigma_k^{(2)}) \leq ( 1- \beta^{k+1}) \left( \sum_{t = 0}^{T-1} \beta^t p^{(1)}(\sigma_t) +  \sum_{t = 0}^{r-1} \frac{\beta^{T + t} p^{(1)}(\gamma_t)}{1 - \beta^r} \right). \nonumber 
 \end{align}

 Taking the $\beta \to 1$ limit on both sides (using L'H\^{o}pital's rule on the right hand side), we get 
 \begin{equation}
    \sum_{t = 0}^{k - 1} p^{(1)}(\sigma_t) + \max_{a^{(1)} \in A^{(1)}} p^{(1)}(a^{(1)}, \sigma_k^{(2)}) < (k+1) \cdot v_{\gamma}^{(1)}. \label{eqtn: rearranged stab with repeats}
 \end{equation}
The strict inequality appears here since the derivative with respect to $\beta$ of the expression obtained by subtracting the left-hand side from the right-hand side in (\ref{eqtn: rearranged stab with repeats}) is strictly positive for large enough $\beta < 1$. So, if we had an equality in the $\beta \to 1$ limit, for large enough $\beta$, $\sigma$ would not be stable. Rearranging the equation then yields the desired result for player $1$ and the proof is analogous for player $2$.
\end{proof}

\nphardness*
\begin{proof}
We prove the result by reducing from the Unbounded Subset-Sum Problem (USSP) with non-negative integer inputs to $\prob$. The (weak) \np-hardness of this problem is well-known (\emph{e.g.}, see~\cite{lueker1975two}). The problem is the following: given  integers $\{b_1, b_2, \ldots, b_n\}$ and $B$, do there exist $c_1, \ldots, c_n \in \mathbb{Z}_{\geq 0}$ such that $\sum_{i = 1}^n c_i b_i = B$? Given an instance of USSP, there is a polynomial time algorithm computing a corresponding instance of $\prob$. Namely, we construct a symmetric game $\Gamma$ with $A^{(1)} = A^{(2)}$ and $|A^{(1)}| = n +2$ where:
\begin{itemize}
    \item $p^{(1)}(a_0, a_0) = B$,
    \item for all $i \in [n]$, $p^{(1)}(a_i, a_i) = B - b_i$,
    \item $p^{(1)}(a_{n+1}, a_0) = 2B - 1$, $p^{(1)}(a_0, a_{n+1}) = - \infty$,
    \item and, for all other action pairs $i \neq j$ and $i = j = n+1$, $p^{(1)}(a_i, a_j) = -\infty$.
\end{itemize}
Define $p^{(2)}(a_i, a_j) = p^{(1)}(a_j, a_i)$ for all $i,j \in \{0\} \cup [n+1]$ so that $\Gamma$ is symmetric.

Now, consider the goal sequence repeating $(a_0, a_0)$ forever. 
Our choice of goal sequence combined with the definition of $\Gamma$ induces hazing costs and threshold values for each action pair $\{(h_{a}, t_{a})\}_{a \in A}$ and goal sequence thresholds $(\theta^{(1)}, \theta^{(2)}) = (B-1, B-1)$. We set parameter $\Delta = B$. The thresholds for all on-diagonal action pairs are trivial by construction, and, using an off-diagonal action pair during hazing would yield infinite hazing. (By Theorem~\ref{thm: folk theorem equil}, there exists a hazing sequence with total hazing at most $B + b_1$ hazing by just repeating $(a_1, a_1)$, so it would even suffice to set the $-\infty$ values to instead be $-B - b_1$.) Then, there existing a hazing sequence hazing each player exactly $B$ is equivalent to finding a sequence of on-diagonal action pairs such that the sum of hazing for each player at the end is $(B,B)$. But, this is then equivalent to solving the instance of USSP, proving that $\prob$ is $\np$-hard.
\end{proof}

\hazingstruct*
\begin{proof}
\emph{(Property 1.)} The total hazing for both players is always non-negative by the stability condition. If that were not the case, a player with negative total hazing could increase their utility by repeatedly restarting the game after the round in which their total hazing became negative. Additionally, note since each action pair in $\gamma$ induces maximum social welfare, we have that, for all $a \in A$, 
\begin{equation*}
p^{(1)}(a) + p^{(2)}(a) \leq  v_{\gamma}^{(1)} + v_{\gamma}^{(2)}.
\end{equation*}
Hence,
\begin{equation} \label{eqtn: nonneg hazing}
 h_a^{(1)} + h_a^{(2)} = v_{\gamma}^{(1)} + v_{\gamma}^{(2)} - p^{(1)}(a) - p^{(2)}(a) \geq 0.   
\end{equation}
As such, the total hazing, $H_k^{(1)} + H_k^{(2)}$, must be increasing with $k$. Therefore, since the minimum hazing stable sequence has total hazing at most $B$, it must always satisfy $H_k^{(1)} + H_k^{(2)} \leq B$.

\emph{(Property 2.)} For all time steps $k_1 < k_2$, 
\begin{equation*}
 H_{k_2}^{(1)} + H_{k_2}^{(2)}  - H_{k_1}^{(1)} -  H_{k_1}^{(2)} = \sum_{t = k_1 + 1}^{k_2} H^{(1)}_{\sigma_t} + H^{(2)}_{\sigma_t} \geq 0, 
\end{equation*}
by~(\ref{eqtn: nonneg hazing}). 

\emph{(Property 3.)} Consider a stable sequence with goal sequence $\gamma$ and minimum total hazing. If there existed time steps $k_1 < k_2$ in the hazing period where  $(H_{k_1}^{(1)}, H_{k_1}^{(2)}) = (H_{k_2}^{(1)}, H_{k_2}^{(2)})$, the modified sequence with actions in time steps $k_1 < t \leq k_2$ removed would also be stable, with goal sequence $\gamma$, and minimum total hazing. Indeed, the stability of actions occurring after time $k_2$ will still hold: the total hazing when reaching those actions is identical among the two sequences.
\end{proof}

\npcompleteness*
\begin{proof}
Suppose there exists a stable sequence with goal sequence $\gamma$ and total hazing at most $B$. By~Lemma \ref{lem: hazing per struct}, there exists a minimum hazing stable sequence with hazing period of length at most $r^2B^2$. Checking stability of a sequence can be done in $\poly(n)$ time, so, this implies the first part of the result. This is also implied by Theorem~\ref{thm:dp}.

For the second part, let $\sigma$ be an optimal stable sequence for $\Gamma$ with only $\poly(n)$ occurrences of action pairs with negative hazing and, without loss of generality (using Lemma~\ref{lem: hazing per struct}), assume that the pairs of total hazing values are distinct for all time steps prior to the end of the hazing sequence. Let $(a^{(1)}, a^{(2)})$ be an action pair occurring in the hazing period with nonnegative hazing for both players. Then, at its first occurrence, the total hazing surpasses the threshold hazing for $(a^{(1)}, a^{(2)})$, by stability of $\sigma$. Since  $(a^{(1)}, a^{(2)})$ hazes both players a non-negative amount, the total hazing after having played $(a^{(1)}, a^{(2)})$ still meets the threshold hazing for $(a^{(1)}, a^{(2)})$. Then, consider the alternative sequence formed by:
\begin{enumerate}
    \item After $(a^{(1)}, a^{(2)})$ is played, repeat it $k - 1$ additional times, where $k$ is the number of times  $(a^{(1)}, a^{(2)})$ originally appeared in the hazing sequence.
    \item  Playing the remainder of the hazing sequence, skipping additional occurrences of  $(a^{(1)}, a^{(2)})$.
\end{enumerate} 
This sequence has same total  hazing as the original sequence by the end of the hazing period. Moreover, the new sequence is also stable: the total hazing at each time step before the first occurrence of $(a^{(1)}, a^{(2)})$ is unchanged, the sequence is stable at each occurrence of $(a^{(1)}, a^{(2)})$, and the total hazing is greater or equal than before at each action pair after all $k$ repetitions of $(a^{(1)}, a^{(2)})$.

\begin{figure}[t]
    \centering
\includegraphics[width=0.6\textwidth]{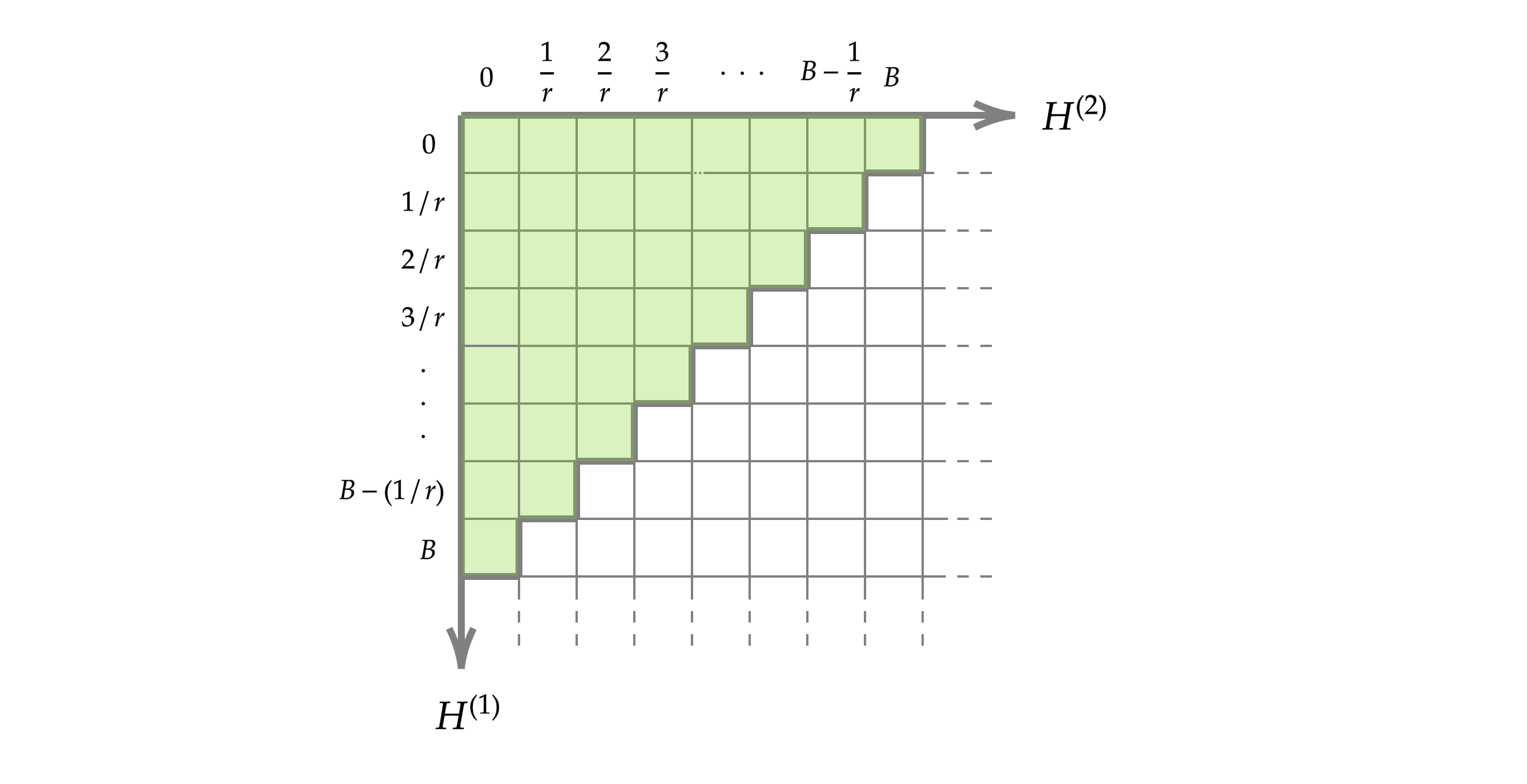}
    \caption{Pictorial representation of $\mem$. The green shading denotes the bounds of hazing value pairs that Algorithm \ref{alg:dp} will search according to Lemma \ref{lem: hazing per struct}.}
    \label{fig:memo}
\end{figure}

Repeating this modification process for each $(a^{(1)}, a^{(2)})$ that appears in the hazing sequence and has non-negative hazing for both players, we can represent the hazing sequence of the resultant optimal stable sequence in a succinct form. That is, the sequence of action pairs paired with the number of times it is repeated consecutively in the sequence will have total size $\poly(n)$. There are $\poly(n)$ total action pairs with non-negative hazing for both players, and each appears only once in this sequence. Then, the remaining action pairs with negative hazing for either player appear a total of $\poly(n)$ times, by assumption. This yields the desired result.
\end{proof}

\thmdp*
\begin{proof}
The proof consists of showing (a) the  correctness, (b) the time complexity, and (c) the space complexity of Algorithm \ref{alg:dp}.
\newline \emph{(a) Correctness:} First, we prove that the algorithm terminates and then that the hazing sequence returned corresponds to a stable sequence with goal sequence $\gamma$ and achieves the minimum sum of total hazing values surpassing the thresholds, $\theta^{(1)}, \theta^{(2)}$. To prove that the algorithm terminates, notice that $\mem$ is a matrix, whose elements belong are all indexed by pairs of total hazing values in $\mathbf{R}^2$. We will prove that pairs of total hazing values between the two players within a candidate hazing sequence can only take on a bounded number of values before the algorithm stops considering that sequence. Specifically, at any given point in the algorithm, $H^{(1)} + H^{(2)} \leq B$ and $\min(H^{(1)}, H^{(2)}) \geq 0$. Notice that since the action pairs in the goal sequence induce maximum social welfare, the sum of hazings per round must be non-negative (see Remark \ref{remark:sumofhazingsperround}). This  implies that once the total sum of hazings becomes greater than $B$, the hazing sequence considered so far cannot reach any total sum hazings that is strictly smaller than $B$, and is therefore suboptimal (see Lemma \ref{lem: hazing per struct}). In Step \ref{step:outofbounds}, Algorithm~\ref{alg:dp} will stop considering that sequence, since $H_{*}^{(1)} +  H_{*}^{(2)} \leq \bd$ by definition. Also, Algorithm~\ref{alg:dp} will stop considering a candidate sequence once that sequence reaches a pair of hazing values that it has reached before (see the second condition of Step~\ref{step:revisited}). Since the total hazing that each player can accumulate is a rational number that is also bounded, indeed the pair of total hazings is in $([0, B] \cap (\frac{1}{r} \cdot \mathbb{Z}))^2$ (see Remark~\ref{remark:rationalhazings}), the while loop (Step~\ref{while}) will loop at most $(r+1)^2B^2$ times.
\par We now show that Algorithm~\ref{alg:dp} will achieve the minimum sum of total hazing values corresponding to a stable sequence with goal sequence $\gamma$. By Lemma~\ref{lem: hazing per struct}, there exists a minimum hazing stable sequence $\sigma$ with properties 1–3. Now, for $k$ in the hazing period of $\sigma$, the total hazing $H_k$ will be added to $\mem$. Indeed, $H_1$ will be added to $Q$ after processing $(0,0)$ since the threshold for $\sigma_1$ is met by $(0,0)$ by stability of $\sigma$ (the first bound of Step~\ref{step:revisited} is not satisfied) and the conditions of Step~\ref{step:outofbounds} are not met by stability and minimality of $\sigma$. Then, by induction, for each $k$ in the hazing period of $\sigma$, $H_k$ will be added to $\mem$. So, the pair of total hazings of $\sigma$ will be added to $\mem$. So, by Step~\ref{opt_threshold}, the total hazing of of the stable sequence returned by Algorithm~\ref{alg:dp}, $\sigma_*$, will be at most the sum of the total hazing of a stable sequence with goal sequence $\gamma$ and minimum sum of total hazings. Stability of the outputted sequence is ensured by Steps~\ref{opt_threshold} and~\ref{step:revisited}.

\emph{(b) Time complexity:} Again, note that the while loop will loop at most $O(r^2B^2)$ times. The steps before the inner for loop take constant time. The inner for loop, Step~\ref{step:innerfor}, iterates $|A|^2$ times per while loop iteration, with each taking constant time. Hence, in total, the while loop takes $O(|A|^2r^2B^2)$ time. Reconstructing the optimal sequence (the loop at Step~\ref{step:reconstruct}) takes $O(r^2B^2)$ time,  since each memo entry is considered only once during reconstruction (the directed graph induced by parent-child relationships among memo entries is acyclic since memo entries are processed only once). Overall the algorithm runs in $O(|A|^2r^2B^2)$ time. 
\newline \emph{(c) Space complexity:} Notice that $\mem$ and $Q$ have at most $O(r^2B^2)$ entries of rational numbers. The optimal sequence is a subset of entries in $\mem$ so reconstructing it requires at most $O(r^2 B^2)$ space.
\end{proof}

\end{document}